\documentclass{article}
\usepackage{amssymb,amsmath}
\usepackage{graphics} 
\usepackage{graphicx} 
\usepackage{url,amsmath,amssymb,amsthm}
\usepackage{etoolbox}
\usepackage{soul,tikz, adjustbox, algorithm, algorithmic}
\usepackage{forest, tikz-qtree,pseudocode,ulem, framed}
\newtheorem{lemma}{Lemma}

\newtheorem{corollary}{Corollary}
\newtheorem{claim}{Claim}

\newenvironment{proof-sketch}{\noindent{\bf Proof Sketch:}\hspace*{1em}}{\qed}
\tikzset{
  treenode/.style = {align=center, inner sep=0pt, text centered,
    font=\sffamily},
  arn_n/.style = {treenode, circle, white, font=\sffamily\bfseries, draw=black,
    fill=black, text width=1.5em},
  arn_r/.style = {treenode, circle, black, draw=black, 
    text width=1.5em, very thick},
  arn_x/.style = {treenode, rectangle, draw=black,
    minimum width=0.5em, minimum height=0.5em}
}

\usepackage{adjustbox,expl3,etoolbox}

\usepackage{amsmath,amsthm} 
\newtheorem{theorem}{Theorem} 
\newtheorem{definition}{Definition} 
 
\newtheorem{example}{Example}

\makeatletter
\providecommand{\institute}[1]{
  \apptocmd{\@author}{\end{tabular}
    \par
    \begin{tabular}[t]{c}
    #1}{}{}
}
\makeatother

\usepackage{authblk}
\begin{document}

\title{It's Not Whom You Know, It's What You (or Your Friends) Can Do: 
Succint Coalitional Frameworks for Network Centralities. }
\author[1,2,*]{Gabriel Istrate}
\author[1,2]{Cosmin Bonchi\c{s}}
\author[1]{Claudiu Gatina}
\affil[1]{Department of Computer Science, West University of Timi\c{s}oara}
\affil[2]{e-Austria Research Institute}
\affil[*]{corresponding author email: gabrielistrate@acm.org}

\maketitle

\begin{abstract}We investigate the representation of game-theoretic measures of network centrality using a framework that blends a social network representation with the succint formalism of cooperative skill games. We discuss the expressiveness of the new framework and highlight some of its advantages, including a  fixed-parameter tractability result for computing
centrality measures under such representations. As an application we introduce new network centrality measures that capture the extent to which neighbors of a certain node can help it complete relevant tasks.
\end{abstract}

\section{Introduction}
Measures of network centrality have a long and rich history in the social sciences \cite{bloch2017centrality} and Artificial Intelligence. Such measures have proved useful for a variety of tasks, such as identifying spreading nodes \cite{suri2008determining} and gatekeepers for information dissemination \cite{narayanam2014shapley}, advertising in multiagent markets \cite{maghami2012identifying}, 
 finding important nodes in terrorist networks \cite{lindelauf2013cooperative,michalak2015defeating}. Recent work has demonstrated that the use of {\it coalitional game-theoretic} versions of centrality measures is especially beneficial \cite{szczepanski2014centrality,tarkowski2017game}, and has motivated the study of other topics, such as the extension of centrality measures to more realistic settings \cite{tarkowski2016closeness}, or the study of (frontiers of) tractability of such measures \cite{aadithya2010game,szczepanski2016efficient,tarkowski2018efficient}. 

The starting point of this paper is the observation that, while motivations for studying many social network concepts (centrality measures  in particular) are often stated informally in terms of \textit{capabilities} that nodes may possess,   \textit{capabilities that could help in performing certain actions}, the actual definitions of such measures do \textbf{not} usually make explicit the different capabilities agent have for acting.

To give just a famous example: Granovetter's celebrated paper on the strength of weak ties \cite{granovetter1977strength} considers edges adjacent to a given node by their frequency of interaction.  It argues that so-called \textit{weak ties} (i.e. to those agents  only interacting with the given node occasionally) are especially important. Such nodes  \textit{may be capable} to \textbf{tell $v$ about a certain job $j$, that $v$ itself does not know about.} The bolded statement may be seen, of course, as specifying a task $tell[j]$, that weak tie neighbors of $v$ may be able to complete as a consequence of their network position.  

The purpose of this paper is \textbf{to study 
 representational frameworks for network centrality that explicitly take into account the acting capabilities of various nodes.} We follow  \cite{skibski2017axiomatic} in advocating the study of network centrality measures from a coalitional game-theoretic perspective. Our concerns are somewhat different: whereas \cite{skibski2017axiomatic} mostly investigated representations of centrality measures from an axiomatic perspective, we study the use of \textit{succint coalitional representation frameworks} \cite{deng1994complexity,conitzer2003complexity,ieong2005marginal,elkind2007computational,bachrach2008coalitional} for such representations. The precise framework we investigate blends a network-based specification $G=(V,E)$ of the agent system  and  the so-called \textit{coalitional skill games} \cite{bachrach2008coalitional}, that informally endow agents with \textit{skills} that may prove instrumental in completing certain \textit{tasks}, and get profit from completing them. A centrality measure arises as an indicator (often the Shapley value) of the ``importance'' of the  agent in the associated coalitional game, measuring the extent to which the agent helps coalitions  profit from completing tasks.

The following is an outline of the (main results of the) paper:  we prove (Theorem~\ref{thm-univ}) that our representations are universal: \textit{all} centrality measures have an equivalent CSG representation. We then identify some limits of this result by identifying a natural property (rationality) that subsumes CSG representations with constant coefficients but not the natural eigenvector centrality measure (Theorem~\ref{counter}). Next we highlight a benefit of CSG representations, in the form of a fixed parameter tractability result (Theorem~\ref{pc}). Also as an  application of our framework, we define two  new centrality measures which aim to measure the extent to which an agent can \textit{enlist its neighbors} to help it complete a set of tasks.  These measures extend some important concepts such as the original game-theoretic network centrality \cite{suri2008determining}. 
We show that our helping measures have tractable explicit formulas for some special CSGs.  We then study (with limited success) the problem of axiomatic characterizations of helping centralities. The paper concludes with brief discussions and open problems. 

For sketches of the missing proofs we refer the reader to the Supplemental Material. 

\section{Preliminaries and Notations} 
\label{prelim}

We will use (and review below) notions from several areas:

\noindent \textbf{Theory of multisets.} A \textit{multiset} is a generalization of a set in which each element appears with a non-negative multiplicity. The \textit{union of two multisets $A,B$}, also denoted $A\cup B$, contains those elements that appear in $A,B$, or both. The multiplicity of such an element in $A\cup B$ is the sum of multiplicities of the element into $A,B$. Given multisets $A,B$, we write $A\subset B$ iff every element with positive multiplicity in $A$ has at least as high a multiplicity in $B$.

\noindent \textbf{Coalitional game theory.} We assume familiarity with the basics of Coalitional Game Theory (see \cite{chalkiadakis2011computational} for a recent readable introduction). For concreteness we review some definitions:

A \textit{coalitional game} is specified by a pair $\Gamma=(N,v)$ where $N=\{1,2,\ldots, n\}$ is a set of \textit{players}, and $v$ is a function $v:2^N\rightarrow \mathbb{R}$, called \textit{characteristic function}, which satisfies $v(\emptyset)=0.$  We will often specify a game by the characteristic function only (since $N$ is implicitly assumed in its definition). Also, denote by $\Gamma(N)$ the set of all coalitional games on $N$. Given integer $0\leq k\leq n$, we denote by $C_{k}(\Gamma)$ the set of coalitions of $\Gamma$ (i.e. subsets of $N$) having cardinality exactly $k$, and let $C(\Gamma)$ be the union of all sets $C_{k}(\Gamma)$. 
A game is \textit{monotonically increasing} if $v$ is monotonically increasing  with respect to inclusion. 

We can represent any game on set $N$ as a linear combination of \textit{veto games}: given $\emptyset \neq S\subseteq N$, the \textit{$S$-veto game on $N$} is the game with characteristic function 
\[
v_{S}(T)=\left\{\begin{array}{ll} 
  1, & \mbox{ if } S\subset T, \\ 
  0, & \mbox{ if otherwise.} \\ 
  \end{array}
\right.
\]
Indeed, it is well-known \cite{karlin2017game} (and easy to prove) that the set of veto games ($v_{S}$) forms a basis for the linear space of coalitional games on $N$. Coefficients $a_{S}$ in the decomposition $v=\sum_{S} a_{S}v_{S}$ are called \textit{Harsanyi's dividends}.  We will use solution concepts associated to coalitional games, notably the \textit{Shapley value}. This index tallies the fraction of the value $v(N)$ of the grand coalition that a given player $x\in N$ could fairly request. It has the formula \cite{chalkiadakis2011computational}  $
Sh[v](x) = \frac{1}{n!} \cdot \sum_{\pi\in S_n} [v(S^{x}_{\pi} \cup \{x\}) - v(S^{x}_{\pi})]$, 
where $S^{x}_{\pi} = \{\pi[ i] | \pi[i] \text{ precedes x in  } \pi\}$.  On the other hand, if $v(S)=\sum_{S} a_{S}v_{S}$ is the veto game decomposition of $v$ then for every $i\in N$ we have 
\begin{equation} 
Sh[v](i)=\sum\limits_{i\in S} \frac{a_{S}}{|S|}
\label{sh-veto}
\end{equation}

We also need to review several particular classes of coalitional games. A \textit{(weighted) dummy game} is a triple $\Gamma=(N,w,v)$ where $w:N\rightarrow [0,\infty)$ is a \textit{weight function} and the characteristic function has the form $v(S)=\sum_{i\in S} w(i)$.  

A \textit{cooperative skill game} (CSG) \cite{bachrach2008coalitional,bachrach2013computing} is a 4-tuple $\Gamma=(N,Sk,T,v)$, where $N$ is a set of \textit{players}, $Sk$ is a set of \textit{skills}, $T$ is a set of \textit{tasks}, and $v$ is a characteristic function. 
We assume that each player $x\in N$ is endowed with a set of skills $Sk_{x}\subseteq Sk$. We extend this notation from players to coalitions by denoting, for every $S\subset N$, $Sk_{S}:= \cup_{x\in S} Sk_{x}$. On the other hand, each task $t\in T$ is identified with a set of skills $T_{t}\subseteq Sk$, the set of skills needed to complete task $T_{t}$. Finally, each task $T_{t}$ has a \textit{profit} $w_{t}\geq 0$. The \textit{value of a coalition} $S\subset N$ is defined as $v(S)=\sum\limits_{t\in T: T_{t}\subseteq Sk_{S}} w_{t}.$ In other words: the value of a coalition $S$ is the sum of profits of all tasks that require only skills possessed by members of $S$. We will actually slightly extend the framework from \cite{bachrach2008coalitional,bachrach2013computing} by requiring that \textbf{tasks are multisets} (rather than sets) \textbf{of skills.} Skillsets are still required to be ordinary sets, but the condition $T_{t}\subseteq Sk_{S}$ is now considered as a multiset inclusion.  A justification for this extension is given by the following example:

\begin{example}\label{ex1} We build upon a scenario from \cite{michalak2015defeating} based on the  9/11 terrorist network initially reconstructed in \cite{krebs2002mapping}. In addition to ordinary nodes (displayed as white circles), some nodes are endowed with one of two skills: M ("martial arts", displayed as yellow squares), P ("pilot", displayed as grey diamonds) (see Figure~\ref{fig911} (a)). A coalition of nodes could execute a hijacking attack iff it contains at least two agents with capability $M$ and one agent with capability $P$. This description maps easily onto an (extended) CSG with a single task, specified as the multiset of skills $\{M,M,P\}$, with profit 1, i.e. a coalition is winning iff it contains at least two M members and at least one P member.\footnote{For ease of exposition/computation we leave out from the specification of our example a condition which was crucial in~\cite{michalak2015defeating}, that an attacking coalition be connected, i.e. that the CSG game is a connectivity game \cite{amer2004connectivity} or a Myerson game \cite{myerson1977graphs}. With additional technical complications along the lines of \cite{skibski2019enumerating} one can probably incorporate this condition into our example as well.} 
\end{example} 

We will denote by $P(s)$ the set of players having a certain skill $s$. 

\textit{Semivalues} \cite{dubey1981value} generalize the well-known concepts of Shapley and Banzhaf index. Given coalitional game $\Gamma$ and $C\in C(\Gamma)$, denote by $MC(C,i):=v(C\cup \{i\})-v(C)$ the \textit{marginal contribution} of player $i$ to coalition $C$. Consider a function $\beta:\{0,\ldots, n-1\}\rightarrow [0,1]$ satisfying 
$\sum_{i=0}^{n-1} \beta(k)=1$. Given semivalue $\beta$, the semivalue $\phi_{i}(v)$ for player $i$ in cooperative game $v$ is $
\phi_{i}(v)=\sum_{k=0}^{n-1} \beta(k)\cdot \mathbb{E}_{C\in C_{k}}[MC(C,i)].$ For $\beta^{Sh}(k)=1/n$ we recover the Shapley value. 
When $\beta^{Ban}(k)=\frac{1}{2^{n-1}}{{n-1}\choose {i} }$ we obtain the Banzhaf index. Another important case is the trivial semivalue $\beta^{triv}(0)=1$, $\beta^{triv}(i)=0$ otherwise. Finally, family of semivalues $\beta=(\beta_{n})$  is called \textit{polynomial time computable} if the two-argument function $(n,k)\rightarrow (\beta_{n})_{k}$ has this complexity.  

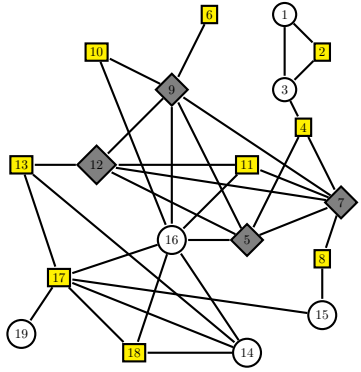
\begin{figure}[h]
\begin{center}
\scalebox{0.50}{

\begin{tikzpicture}[shorten >=1pt, auto, node distance=3cm, ultra thick]
    \tikzstyle{node_ordinar} = [circle, draw=black,fill=white!]
    \tikzstyle{node_pilot} = [diamond, draw=black, fill=gray!]
    \tikzstyle{node_karatist} = [rectangle, draw=black, fill=yellow!]

    \tikzstyle{edge_style} = [draw=black, line width=2, ultra thick]
    \node[node_ordinar] (v1) at (3,5) {1};
    \node[node_ordinar] (v3) at (3,3) {3};
    \node[node_ordinar] (v14) at (2,-4) {14};
    \node[node_ordinar] (v15) at (4,-3) {15};
    \node[node_ordinar] (v16) at (0,-1) {16};
    \node[node_ordinar] (v19) at (-4,-3.5) {19};
    
    \node[node_pilot] (v5) at (2,-1) {5};
    \node[node_pilot] (v7) at (4.5,0) {7};
    \node[node_pilot] (v9) at (0,3) {9};
    \node[node_pilot] (v12) at (-2,1) {12};
    
    \node[node_karatist] (v2) at (4, 4) {2};
    \node[node_karatist] (v4) at (3.5,2) {4};
    \node[node_karatist] (v6) at (1,5) {6};
    \node[node_karatist] (v8) at (4,-1.5) {8};
    \node[node_karatist] (v10) at (-2,4) {10};
    \node[node_karatist] (v11) at (2,1) {11};
    \node[node_karatist] (v13) at (-4,1) {13};
    \node[node_karatist] (v17) at (-3,-2) {17};
    \node[node_karatist] (v18) at (-1,-4) {18};

    \draw[edge_style]  (v1) edge (v2);
    \draw[edge_style]  (v1) edge (v3);
    \draw[edge_style]  (v2) edge (v3);
    \draw[edge_style]  (v3) edge (v4);
    \draw[edge_style]  (v4) edge (v5);
    \draw[edge_style]  (v4) edge (v7);
    \draw[edge_style]  (v5) edge (v7);
    \draw[edge_style]  (v5) edge (v9);
    \draw[edge_style]  (v5) edge (v12);
    \draw[edge_style]  (v5) edge (v16);
    \draw[edge_style]  (v6) edge (v9);
    \draw[edge_style]  (v7) edge (v8);
    \draw[edge_style]  (v7) edge (v9);
    \draw[edge_style]  (v7) edge (v11);
    \draw[edge_style]  (v7) edge (v12);
    \draw[edge_style]  (v8) edge (v15);
    \draw[edge_style]  (v9) edge (v10);
    \draw[edge_style]  (v9) edge (v12);
    \draw[edge_style]  (v9) edge (v16);
    \draw[edge_style]  (v10) edge (v16);
    \draw[edge_style]  (v11) edge (v12);
    \draw[edge_style]  (v11) edge (v16);
    \draw[edge_style]  (v12) edge (v13);
    \draw[edge_style]  (v13) edge (v14);
    \draw[edge_style]  (v13) edge (v17);
    \draw[edge_style]  (v14) edge (v16);
    \draw[edge_style]  (v14) edge (v17);
    \draw[edge_style]  (v14) edge (v18);
    \draw[edge_style]  (v15) edge (v17);
    \draw[edge_style]  (v16) edge (v17);
    \draw[edge_style]  (v16) edge (v18);
    \draw[edge_style]  (v17) edge (v18);
    \draw[edge_style]  (v17) edge (v19);

    \end{tikzpicture}}
\end{center}
\caption{The 9/11 WTC attack social nework (after \cite{krebs2002mapping}, with skills assigned by \cite{michalak2015defeating}).}
\label{fig911}
\end{figure} 

We note the following very simple result:  
 \begin{lemma}  
 If $\mbox{  }\Gamma=(N,w,v)$ is a weighted dummy game and $\beta$ is a semivalue then for every $i\in N$, $\phi_{i}(v)=w(i)$. 
\end{lemma}

\noindent\textbf{Graph Theory and Network Centralities.} 
A \textit{graph} is a pair $G=(V,E)$ with $V$ a set of \textit{vertices} and $E$ a set of \textit{edges}. The degree of $v$, $deg(v)$ is the number of nodes $v$ is connected to by edges. We will use $\Delta$ to denote the maximum degree of a node in $V$. If $v\in V$ is a vertex we will denote by $N(v)$ the set of neighbors of $v$ in $G$ and by $\hat{N}(v)=N(v)\cup \{v\}$.  We extend these definitions to sets $S\subseteq V$ by $N(S)=\{z\mbox{ }|\mbox{ }\exists w\in S, (z,w)\in E\}$. 

We will denote by $\mathcal{G}^{V}$ the set of all graphs on the vertex set $V$.  A \textit{centrality index} is a function, $c : \mathcal{G}^V \rightarrow \mathbb{R}^V$  that assigns to every node $v \in V$ a real number, called \textit{the centrality of $v$} quantifying the importance of node $v$ in $G$. We will denote by $\mathcal{C}^{V}$ the set of all centrality measures on the set $V$.  We will usually drop $V$ from our notation and write $\mathcal{G}, \mathcal{C}, \ldots$  instead of $\mathcal{G}^V, \mathcal{C}^{V}$ and so forth. 

We will use several concrete measures of centrality. The following is a listing of some of them: 
\begin{itemize} 
\item[-] \textit{degree centrality} of node $v$ in graph $G$ is defined by $c^{D}(v,G)=|\{(v,u)\in E | u\in E\}|$. 
\item[-] \textit{betweenness centrality} of node $v$ in graph $G$ is defined as follows: given two distinct nodes $z_{1},z_{2}\in V$, denote by $p(z_{1},z_{2})$ the number of shortest paths in $G$ between $z_{1}$ and $z_{2}$, and by $p(z_{1},z_{2},v)$ the number of shortest paths \textit{passing through $v$}. Now we can define betweenness centrality as 
$c^{close}(v,G)=\sum_{z_{1}\neq z_{2}\in V} \frac{p(z_{1},z_{2},v)}{p(z_{1},z_{2})}$. 
\item[-] \textit{game-theoretic network centrality} of node $x$ in graph $G$ is defined as the Shapley value of $x$ in game $\Gamma$ with characteristic function $v_{*}(S)=|S\cup N(S)|.$
\item[-] the \textit{eigenvector centrality} of node $v$ in graph $G$ is defined as the $v$'th component of the eigenvector associated to the largest eigenvalue of the adjacency matrix of $G$. 
\end{itemize} 

\noindent\textbf{Coalitional Network Centralities.}
Following \cite{skibski2017axiomatic}, a \textit{representation function} is a function $\psi$ mapping every graph $G=(V,E)$ onto a cooperative game $\Gamma_{G}$ whose players are the vertices of $G$, $\Gamma_{G}=(V,v_{G})$. We will call a representation w \textit{skill-based} if for every graph $G$, the associated game $\Gamma_{G}$ is a CSG.  A \textit{coalitional centrality measure} is a pair $(\psi,\phi)$, where 
$\psi$ is a representing function and $\phi$ is a solution concept. A \textit{skill-based centrality measure} is one for which representation $\psi$ is skill-based. Given semivalue $\beta$, a \textit{$\beta$-skill-based centrality measure} is a pair $(\psi,\phi)$ where $\psi$ is a skill-based and $\phi$ is the semivalue induced by $\beta$.  A skill-based centrality measure is \textit{trivial} iff the solution concept $\phi$ is simply the value function of the CSG associated to graph $G$ by $\phi$. Note that for weighted dummy games this is equivalent to requiring that $\phi$ is the semivalue induced by the trivial semivalue $\beta^{triv}$. 

\noindent\textbf{Parameterized Complexity.} A \textit{parameterized problem} is specified by a set of pairs $W\subseteq \Sigma^{*}\times  \mathbb{N}$ and a function $f:A\rightarrow \mathbb{N}$. Problem $(W,f)$ is \textit{fixed-parameter tractable} if there exists a computable function $g:\mathbb{N}\rightarrow \mathbb{N}$, an integer $r>0$ and an algorithm $A$ that computes $f(z)$ on inputs $z=(x,k)$ from $W$ in time $O(g(k)\cdot |x|^{r})$.

\section{Universality of Skill-Based Centralities.} 
\label{univ} 

\cite{skibski2017axiomatic} have shown that any centrality measure is equivalent to a coalitional centrality measure. We make this result slightly more precise: the cooperative game can be taken to be a CSG and the solution concept can be induced by any arbitrary semivalue:  

\begin{theorem} 
For every semivalue $\beta$ and every centrality $c\in \mathcal{C}$ there exists an equivalent $\beta$-skill-based  representation. 
\label{thm-univ}
\end{theorem} 

\begin{proof} 
Let  $c\in \mathcal{C}$ be a centrality measure on graph $G$. Consider the \textit{dummy game} in which $v(S)=\sum_{v\in S} c(v)$. 

One can represent this dummy game by associating to $G$ the CSG game $\Gamma$ as follows: $Sk=V$, (i.e. skills correspond to agents). For every $v\in V$ we define $w(v)=c(v)$. Finally $T=V$. This yields a skill-based representation $\psi_{C}$. Completing this representation by the trivial semivalue induced by $\beta$ in $\Gamma$ yields a $\beta$-skill-based centrality measure which (by Lemma 1) is easily seen to be equivalent to $c$. 
\end{proof} 

\begin{example} \textbf{[Degree centrality:]}
Consider a graph $G=(V,E)$. We associate to $G$ a game $\Gamma$ as follows: skills correspond to edges of $G$. A node of $G$ has a skill $e$ iff it is incident with $e$. Tasks correspond to edges as well. 
\end{example} 

Sometimes, as the following example shows, the "natural" representation of centralities using CSG is inefficient, as the number of tasks may be exponential in the size of graph $G$.

\begin{example} \textbf{[Betweenness centrality]:}  
Consider a graph $G=(V,E)$. Associate to $G$ a game $\Gamma$ as follows: skills correspond to edges of $G$. A node of $G$ has a skill $e$ iff it is incident with the corresponding edge. Tasks correspond to shortest paths connecting two nodes, say $z_{1},z_{2}$ in $G$. Such a task has weight equal to the inverse of the number of shortest paths between $z_{1},z_{2}$. The trivial skill-based centrality measure coincides with (ordinary) betweenness centrality. 
\label{bc}
\end{example}

\subsection{Limitations of (Rational) Network Centralities.} 

In Theorem~\ref{thm-univ} we were, in some sense, "cheating", as the values of network centrality were built in the weights the dummy game representing the measure. In particular this game depended on the graph $G$, not only on $n$, the number of vertices. It is natural to ask whether universality fails  once we impose some further restrictions on the framework that precludes such "pathological" representations. 


In the sequel we study an interesting and natural restriction on characteristic functions and centrality measures: that they are "rational functions of the graph topology", i.e. a quotient of two polynomials. 
We formalize this idea as follows: Given a set of vertices $V$, denote by $E_{n}(V)$ the set of subsets 
$w=\{v_{1},v_{2}\}$ of distinct vertices in $V$. Associate to every $w\in E_n(V)$ a boolean variable $X_{w}$. We can interpret the set $E$ of edges of any graph $G$ on $V$ as a 0/1 assignment $(X_{w})_{w\in E_{n}(V)}\in \{0,1\}^{E_{n}(V)}$,   $X_{w}=1$ if $w\in E$, $X_{w}=0$, 
otherwise.  By forcing notation we will write $E$ instead of $(X_{w})_{w\in E_{n}(V)}\in \{0,1\}^{E_{n}(V)}$. We do similarly for vertices, identifying a vertex $i$ with a boolean variable $Y_{v}$. 
This way we can specify a set of vertices $S$ by a boolean vector, corresponding to those vertices $v$ with $Y_{v}=1$. 

\begin{definition} 
A family of characteristic functions $(v_{n})_{n\geq 1}$ is called {\rm rational} if there exists two families of polynomials 
$P_{n}(X_{e},Y_{v})$ and $Q_{n}(X_{e},Y_{v}) \in \mathbb{Q}[X,Y]$ such that for every
$n\geq 1$ and $S\subseteq [n]$
\begin{equation} 
v_{n}(S)= \frac{P_{n}[E,S]}{Q_{n}[E,S]}
\end{equation}  
\end{definition}

\begin{example} 
For every $S\subseteq [n]$ characteristic functions $v_{S}$ are rational. Indeed $v_{S}(\cdot)=\prod_{i\in S} Y_{i}.$
\label{veto-rat}
\end{example} 

\begin{example} Characteristic function $v_{*}$ in the definition of game-theoretic network centrality is rational. 
Indeed, 
\[
v_{*}(S)=\sum_{i\in [n]} Y_{i}+ \sum_{i\in [n]} (1-Y_{i})[1-\prod_{j\neq i} (1- Y_{j}X_{i,j})]
\]

To see that this equality is true: each $i\in S$ contributes 1 to the first sum. Only $i\not \in S$ may contribute to the second sum, but only when some term $1-Y_{i}X_{i,j}$ is equal to zero, that is when there is some $j\neq i$ with $Y_{j}=1$ (i.e. $j\in S$) and $X_{i,j}=1$ (i.e. $(i,j)\in E$). 
\end{example}

\begin{definition} 
Let $V$ be a set of vertices. 
A centrality measure $c=(c_{n})_{n\geq 1}$ is \textrm{rational} iff there exist multivariate polynomials $P_{n,v},Q_{n}\in \mathbb{Q}[X]$  such that, for every $n\geq 1$, and every graph $G=(V,E)$ on $V$ we have 
\begin{equation} 
c_{n}(v,G)=\frac{P_{n,v}[E]}{Q_{n}[E]}
\end{equation} 
\end{definition} 

\begin{example} 
Degree centrality is rational. Indeed, one may take $P_{n,v}[X]=\sum\limits_{e\ni v}X_{e}$ 
and $Q_{n}[X]=1.$
\end{example}

The case of betweenness centrality is  more interesting: 

\begin{theorem} 
Betweenness centrality is rational. 
\end{theorem} 
\begin{proof} 
For every set of vertices $V$, define $\mathcal{D}$ to be the family of simple paths in the complete graph on $V$. Given $P\in \mathcal{D}$, define monomials $X_{P}=\prod_{e\in P} X_{e}$
and $
\tilde{X}_{P}= X_{P}\cdot \prod\limits_{\stackrel{Q\in \mathcal{D}}{|Q|<|P|}}(1-X_{Q}). $
 Also, for vertices $z_{1}\neq z_{2}$ define $
P_{n,v,z_{1},z_{2}}[X]=\sum\limits_{P:z_{1}\rightarrow v\rightarrow z_{2}} \tilde{X}_{P}$, $Q_{z_{1},z_{2}}[X]=\sum\limits_{P:z_{1}\rightarrow z_{2}} \tilde{X}_{P}.$
With these notations, we claim that we have the following formula 
\begin{equation}
BC_{n}[v]=\sum_{v_{1}\neq v_{2}\in V} \frac{P_{n,v,z_{1},z_{2}}[X]}{Q_{z_{1},z_{2}}[X]} 
\label{bc2}
\end{equation} 
To prove equation~(\ref{bc2}) we first show that 
\begin{claim} 
Given graph $G=(V,E)$, $\tilde{X}_{P}=1$ iff $P$ is a shortest path in $G$. 
\label{foo} 
\end{claim}
\begin{proof}
$\tilde{X}_{P}=1$ iff $X_{P}=1$ and for all $Q\in \mathcal{D}$, $|Q|<|P|$, $X_{Q}=0$, that is $Q$ is {\bf not} a path in $G$. 
\end{proof} 
Applying Claim~\ref{foo} we infer that $P_{n,v,z_{1},z_{2}}$ count shortest paths between $z_{1},z_{2}$ passing through $v$ and $Q_{z_{1},z_{2}}$ counts all shortest paths between $z_{1},z_{2}$. 
\end{proof}

The following two theorems show that the family of rational centrality measures is reasonably comprehensive:

\begin{theorem} 
Every centrality measure induced by a rational family of characteristic functions is rational. 
\end{theorem} 
\begin{proof}
From the marginal contribution formula for the Shapley value.  
\end{proof} 

\begin{corollary}
Game-theoretic network centrality \cite{michalak2013efficient} is a rational centrality measure.  
\end{corollary} 

\begin{corollary}
If $v_{n}$ is a family of characteristic functions whose Harsanyi dividends are rational numbers then the family of centrality measures induced by $v_{n}$ is rational.  
\end{corollary}
\begin{proof} 
From Exp.~\ref{veto-rat} and the fact that  linear combinations of rational functions with coefficients in $\mathbb{Q}$ are rational. 
\end{proof}

\begin{theorem}
Every centrality measure induced by a family of CSG with constant coefficients in $\mathbb{Q}$ is rational. 
\label{rational-csg}
\end{theorem} 

Theorem~\ref{rational-csg} connects rationality to the representation of centrality by  CSG, essentially showing that when we disallow adaptive representations like those used in the proof of Theorem~\ref{thm-univ} induced centrality measures are indeed rational.

Despite these two results, rational measures are \textbf{not} universal; they fail to capture a natural centrality measure: 

\begin{theorem} Eigenvector centrality is not rational. 
\label{counter} 
\end{theorem}
\begin{proof} 
Consider the graph $G$ from Figure~\ref{p2} (a). Simple computations show that eigenvector centrality values of node $1,3$ are equal to $\frac{1}{\sqrt{2}+1}$ and that of $2$ is equal to $\frac{\sqrt{2}}{\sqrt{2}+1}$. In particular the eigenvector centrality of $1,3$ is an irrational number. But this would be impossible if the eigenvector centrality were a rational centrality measure. 
\end{proof} 

It would be interesting to define an extension of the family of rational centralities that captures all "natural" centralities. 

\section{Parameterized Complexity of Computing Skill-Based Centralities. }
\label{param}

Since computing the Shapley value of CSG is $\#P$-complete \cite{aziz2009algorithmic}, it follows that computing skill-based centralities is intractable in general. On the other hand, by imposing a natural restriction on the family of CSG games under consideration, that of an existence of an upper bound on the largest set of skills needed for a task, we get a fixed-parameter tractable class of algorithms:

\begin{theorem} 
 Let $\beta$ be a poly-time computable family of semivalues. The following problem 
\begin{itemize} 
\item[- ][INPUT:] A CSG $\Gamma=(N,v)$ and a player $i\in N$.
\item[-][TO COMPUTE:] Semivalue $\Phi_{i}^{\beta}(v).$
\end{itemize} 
parameterized by $k$, the cardinality of the largest skill set required by any task, is fixed parameter tractable. 
\label{pc} 
\end{theorem}

\section{An Application: Helping Centralities.}
\label{helping} 

In this section we give an application of the idea of representing network centralities by CSG. \cite{alshebli2019measure} have recently defined (using different ideas) a centrality measure that quantifies the extent to which a given agent adds value to a group. On the other hand, an agent may be valuable to a group even when it lacks the skills to contribute to completing a given task, provided it is capable to \textit{enlist neighbors with such skills}. 

\begin{example} 
Members of the program committee for computer science conferences often use subreviewers to referee papers. Each paper needs to receive a minimal number (say three) of reviews. A PC member may lack the \textbf{skill} to competently review the paper itself. But the ability it may have to \textrm{help the reviewing process} by enlisting subreviewers with the required reviewing skills, in order to complete the \textbf{task} of getting three reviews for the given paper, is highly valuable. 
\end{example}

\begin{example} 
Consider again the coalitional game-theoretic framework for the WTC 9/11 terrorist network in Example~\ref{ex1}. Nodes $4,7,9,11,12,13,16$ could have assembled an attacking team consisting of (some of) their neighbors. In the case of node $16$ (N. Alhazmi) this happens despite not being known to have had any of the two required skills $P,M$. Because of this fact, node 16 intuitively can "help" all non-winning coalitions (which may already include it !) by enlisting its neighbors. This is intuitively,  not true for nodes $4,7,9,11,12,13:$ they do not "help" those coalitions that were already turned into winning coalitions by their mere joining. Neither do any other nodes in the network. So, intuitively, node 16 should be the "most helping" node. 
\end{example} 

It turns out that properly defining such a centrality notion is somewhat subtle and may not have a single, always best solution. The following could be the natural first idea: 

\begin{definition} Given coalitional game $\Gamma$ and graph $G$, we define \textit{the centrality extension of $\Gamma$ on graph $G$} as the game $\Gamma_{1}$ with value function $v_{cen}(S)=v(S\cup N(S))$. 
\end{definition} 

\begin{example} 
Consider the dummy game $v(S)=|S|$. Then for every graph $G$ the Shapley value of the centrality extension $v_{cen}(S)$ of $v$ on graph $G$ is exactly the \textit{game-theoretic centrality} of $G$ \cite{suri2008determining}. 
\label{gtcentral}
\end{example} 

Since we want to disentangle a player's capability to help solve tasks itself from its capability to get help from its neighbors, defining helping centrality would require subtracting the Shapley value of $v$ from the Shapley value of its centrality extension $v_{cen}$. However, this  is problematic: since $v(N)=v_{cen}(N)$ and the Shapley values of $v_{cen},v$ add up to $v(N)$, quantity $Sh[v_{cen}](i)-Sh[v](i)$ will be negative for some players $i$ (unless $Sh[v]=Sh[v_{cen}]$,  in which case the centrality would be identically zero) ! 

A variation of this  approach works, however, sometimes: instead of subtracting $Sh[v]$, only subtract a "scaled down" version of this quantity. That is, define helping centrality as (a multiple of) $Sh[v_{cen}]-\alpha Sh[v]$ for some  suitable $\alpha \in (0,1)$. An appropriate choice seems to be the following
\begin{definition} 
For CSG game $\Gamma$ on graph $G$ define
\[
\Delta^{*}=\max_{\stackrel{t\in T}{s\in T_{t}}}\{|N(C)\setminus C|: C\subseteq P(s), |C|=n_{s}-k_{s,t}+1\}
\]
where $n_{s}$ is the number of players having skill $s$ and $k_{s,t}$ is the number of copies of skill $s$ needed to accomplish task $t$. 
\end{definition} 
Note that $\Delta^{*}=\Delta$ in games where $n_{s}=k_{s,t}$ for every task $t\in T$ and $s\in T_{t}$, in particular for games where each skill is possessed by a single player. 

\begin{definition} 
Define the \textit{Shapley-based helping centrality} of a node $x\in V$ in a game $\Gamma=(N,v)$ on graph $G$ by 
\begin{align}
HC(x)=(1+\frac{1}{\Delta^{*}})[Sh[v_{cen}](x)-\frac{Sh[v](x)}{\Delta^{*}+1}]= \nonumber \\ = (1+\frac{1}{\Delta^{*}})Sh[v_{cen}](x)- \frac{1}{\Delta^{*}}Sh[v](x)
\end{align}
\end{definition} 

This definition is sensible in at least the following three settings (Theorem~\ref{one}, stated later, offers yet another one):

\begin{theorem} 
Consider a CSG game $\Gamma=(V,v)$ in which every player $x$ possesses at most one skill. Then for every player $x$, $HC(x)\geq 0$ and $\sum\limits_{y\in V} HC(y)=1$. 
\label{positive-helping} 
\end{theorem} 
\noindent In the next two examples each skill is unique to some player: 
\begin{example}
Let $G=S_{n}$ be the star graph with $n$ vertices (Figure~\ref{p2} (b)) and let $S=\{2,3,\ldots, n\}$. Then 
\[
v_{cen,S}(T)=\left\{\begin{array}{ll} 
  1, & \mbox{ if } \{2,3,\ldots n\} \subseteq T \mbox{ or }1\in T, \\ 
  0, & \mbox{ if otherwise.} \\ 
  \end{array}
\right.
\]
Simple computations yield $Sh[v_{S}](1)=0$, $Sh[v_{S}](i)=1/(n-1)$ for $i=2,\ldots n$. Also 
$Sh[v_{cen,S}](1)=1-\frac{1}{n}$, $Sh[v_{cen,S}](i)=\frac{1}{n(n-1)}$ for $i=2,\ldots n$. In conclusion $
HC(1)=1, HC(i)=0,\mbox{ for }i=2,\ldots n.$
Node $1$ is the only one that has positive helping centrality, despite being a null player ! Also, sensibly, the helping centrality of all other nodes is zero, as their only neighbor is 1 which is not in the coalition $S$, so they cannot help. 

\end{example}

\begin{example} In the setting of Example~\ref{gtcentral}, 
the Shapley value of node $x$ has the formula \cite{michalak2013efficient} $
Sh[v](x)=\sum_{y\in \hat{N}(x)}\frac{1}{deg(y)+1}.$
Hence $HC(x)=(1+\frac{1}{\Delta})[ \sum\limits_{y\in \hat{N}(x)}\frac{1}{deg(y)+1} - \frac{1}{\Delta+1}]$, which is $\geq 0$, since $x\in \hat{N}(x)$ and $deg(x)\leq \Delta$. 
\end{example} 

However, it is not clear that, as defined, $HC[v](x)$ is nonnegative in all CSG. 
This motivates giving a second definition, which departs from the idea of giving a notion of helping centrality based on the Shapley value: 

\begin{definition} 
Given coalitional game $\Gamma$ and graph $G$, we define \textit{the Helping Shapley value of a player $x$} by $
HSh[v](x) = \frac{1}{n!} \cdot \sum_{\pi\in S_n} [v(S^{x}_{\pi} \cup \{x\}\cup N(x)) - v(S^{x}_{\pi})],$ 
$\mbox{ where }S^{x}_{\pi} = \{\pi[ i] | \pi[i] \text{ precedes x in  } \pi\}$. The \textit{ helping centrality} of player $x$ is defined as 
\begin{align} 
Help(x)=HSh[v](x)-Sh[v](x) = \nonumber \\  = \frac{1}{n!} \sum_{\pi\in S_n} [v(S^{x}_{\pi} \cup \{x\}\cup N(x)) - v(S^{x}_{\pi} \cup \{x\})]. 
\end{align} 
 Note that if game $\Gamma$ is monotonically increasing then $v(S^{x}_{\pi} \cup \{x\}\cup N(x)) - v(S^{x}_{\pi} \cup \{x\})\geq 0$. When the sign is strictly positive say that $x$ \textit{helps ordered coalition $S^{x}_{\pi}$.} Finally, by performing appropriate divisions we also consider the \textit{normalized} versions $\widetilde{HSh}$ and $\widetilde{Help}$.
\end{definition}  

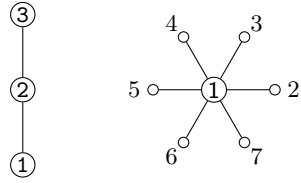
\begin{figure} 
\begin{center} 
\begin{minipage}{.10\textwidth}
\begin{center} 
\begin{tikzpicture}[scale=0.10,font=\sffamily,every path/.style={>=latex},every node/.style={draw,circle,inner sep=0.7pt,font=\small\sffamily}]
  \node(q1) at (0, 0){$\mathtt{1}$};
  \node(q2) at (0, 10)  {$\mathtt{2}$};
  \node(q3) at (0, 20) {$\mathtt{3}$};
 
  \draw (q1) to (q2);
  \draw (q2) to (q3);
 
\end{tikzpicture}
\end{center} 
\end{minipage} 
\begin{minipage}{.30\textwidth}
\begin{center}
\begin{tikzpicture}[scale=0.10,font=\sffamily,every path/.style={>=latex},every node/.style={draw,circle,inner sep=0.7pt,font=\small\sffamily}]
  
  \def \n {6}
  
  \foreach \x [count=\xi from 2] in {1,...,\n}{
    \pgfmathparse{\x * (360 / \n) - (360 / \n)};
    \node[circle,inner sep=1.4pt,label=\pgfmathresult:{$\xi$}] (N-\x) at (\pgfmathresult:8.1cm){};
  };
  
  \node (N-0) at (0:0cm){1}; 
  
  \foreach \x in {1,...,\n}{%
        \path (N-0) edge [-] (N-\x);
  }
\end{tikzpicture}
\end{center} 
\end{minipage}
\end{center} 
\caption{(a). Graph $P_{2}$ (b). Star graph $S_{7}$. }
\label{p2}
\end{figure}

\begin{example}
Let $G=P_{2}$ be the graph in Figure~\ref{p2}, and $\Gamma$ be the \textit{$T$-veto game} 
corresponding to coalition $T=\{1,3\}$. That is 
$v_{\{1,3\}}(S)=\left\{\begin{array}{ll} 
  1, & \mbox{ if } \{1,3\}\subset S, \\ 
  0, & \mbox{ if otherwise.} \\ 
  \end{array}
\right.$
Considering the permutations in $S_{3}$ in the order $(1,2,3)$, $(1,3,2)$, $(2,1,3)$, $(2,3,1)$, $(3,1,2)$, $(3,2,1)$, simple  computations show that 
$HSh[v_{\{1,3\}}](1)=\frac{1}{6}(0+0+0+1+1+1)=\frac{1}{2}$, $HSh[v_{\{1,3\}}](2)=\frac{1}{6}(1+0+1+1+0+1)=\frac{2}{3}$, $HSh[v_{\{1,3\}}](3)=\frac{1}{6}(1+1+1+0+0+0)=\frac{1}{2}$. As for helping centralities, we have $Sh[v_{\{1,3\}}](1)=Sh[v_{\{1,3\}}](3)=\frac{1}{2}$ and $Sh[v_{\{1,3\}}](2)=0$, so $Help(1)=Help(3)=0$, and $Help(2)=2/3$. Node $2$ is the only one that has positive helping centrality, despite being a null player !
\end{example}

A potential disadvantage of  measures $HSh$ and  $Help$ is that they do not have easy interpretations in terms of classical notions of coalitional game theory. On the other hand, they have exact formulas somewhat reminiscent of the corresponding formula~(\ref{sh-veto}) for the Shapley value:

\begin{theorem} Given game $\Gamma=(N,v)$ with veto decomposition $v=\sum_{S} a_{S}v_{S}$, graph $G=(N,E)$, $i\in N$ and $S\subseteq N$ denote by $NC(i,S)$ the set of nodes in $S\setminus \hat{N}(i)$ and by $Cov(i,S)=S\cap \hat{N}(i)$. Then $
HSh[v](i)=\sum\limits_{\stackrel{\emptyset \neq S}{i\in S}} \frac{a_{S}}{|NC(i,S)|+1}+\sum\limits_{\stackrel{\emptyset \neq S}{i\in N(S)\setminus S}} \frac{a_{S}}{|NC(i,S)|+1}[1-\frac{1}{|Cov(i,S)|}]$ and $Help(i)=\sum\limits_{\stackrel{\emptyset \neq S}{i\in S}} \frac{a_{S}(|Cov(i,S)|-1)}{|S|(|NC(i,S)|+1)}+  \sum\limits_{\stackrel{\emptyset \neq S}{i\in N(S)\setminus S}} \frac{a_{S}[|Cov(i,S)|-1]}{(|NC(i,S)|+1)|Cov(i,S)|}.$
\label{seven}
\end{theorem} 

Unfortunately a formula similar to~(\ref{sh-veto}) for $Sh[v_{cen}]$ (and ultimately for $HC[v]$) seems hard to obtain. The reason is that in order to compute $Sh[v_{cen}]$ we would need to compute the Shapley values of games of type 
\[
v_{cen,S}(T)=\left\{\begin{array}{ll} 
  1, & \mbox{ if } S\subset T\cup N(T), \\ 
  0, & \mbox{ if otherwise.} \\ 
  \end{array}
\right.
\]
Doing this requires computing the probability that a random set of vertices is a vertex cover, 
which seems infeasible. This shows another problem of the Shapley-based helping centrality. 

\subsection{Helping Centralities in Terrorist Networks\thanks{Some details on the (simple) computations substantiating our claims in this section are given as Supplemental Material.}}
 
We next apply our helping centrality measures to the terrorist network in Figure~\ref{fig911}. We could estimate helping centralities using sampling techniques similar to those for the Shapley value, but in this case of the 9/11 network exact computations are actually feasible. To reduce overhead, call two nodes equivalent iff (a). They have the same set of skills (b). The families of multisets of skills of their neighbors are identical as multisets. For Figure~\ref{fig911} this relation splits $V$ into the following equivalence classes: 
$\{1,19\},\{2\},\{3,15\},\{18\},\{17\},\{4,11\},\{5\}$, $\{16\},\{14\}$, $\{13\},\{9,12\},\{6,8,10\},\{7\}$. It is easy to see that equivalent nodes have identical Helping and Shapley-Based Helping Centralities. Nodes 2 helps no ordered coalitions, hence $Help(2)=0$. At the other extreme, 16 has the highest normalized helping centrality, $0.126117$. The complete node ordering by decreasing Helping Centrality is $16> 13> \{7,9,12\}> \{6,8,10,11,4\}>5>  \{3,14,15\}> \{1,19\} >\{17,18\}>2.$


As for the Shapley-Based Helping Centrality, the 9/11 Network falls within the scope of Theorem~\ref{positive-helping}. $\Delta^{*}=11$, as witnessed by the coalition $C$ of $M$ nodes except $11$. Again 16 has the highest value, $HC(16)= 0.118687$, while the order by decreasing $HC$ is $16> 13> \{7,9,12\}> \{6,8,10,11\}>4> 5>  \{3,14,15,1,19\}>\{17,18,2\}.$ Both measures  identify node 16 as the most helpful, with comparable centralities, and give quite similar orderings (an interesting fact, since the two measures were fairly different). The ordering produced by $Help$ seems slightly more discriminating. 

\section{Complexity of Helping Centralities.} 
As expected, computing helping centralities for arbitrary CSG is computationally intractable: 

\begin{theorem} 
The following problems are $\#P$-complete: 
\begin{itemize} 
\item[-] [INPUT:] Graph $G$, CSG $\Gamma$,  and player $x\in N$.
\item[-] [COMPUTE:] (a). The Helping Shapley value $HSh[v](x)$. (b). The Shapley-based helping centrality $HC(x)$ of node $x$. 
\end{itemize} 
\label{sharp-p}
\end{theorem}

\subsection{A Tractable Special Case: Pure  Skill Games} 
In the sequel we highlight a special class of CSG for which computing helping centralities is tractable. 

\begin{definition} 
A \textit{pure skill game} is a CSG where, for every $t\in T$, $|T_{t}|=1$ (every task presumes a single skill). 
\end{definition}



\begin{theorem} For pure skill games $v$ and player $x\in V$
\begin{align}
& HC(x) = \sum\limits_{t\in T_{\hat{N}(x)}} w_{t}\frac{(\Delta^{*}+1)|P(t)|-|P(t)\cup N(P(t))|}{\Delta^{*}|P(t)||P(t)\cup N(P(t))|} \nonumber\\
& Help(x)= \sum\limits_{t\in T_{x}} \frac{w_{t}}{|P(t)|(|P(t)|+1)}\nonumber
\end{align} 
We have denoted by $P(t)$ the set of players that have the unique skill needed to complete task $t$. Consequently, both quantities $Help(x)$ and $HC(x)$ are $\geq 0$ for all players $x$. 
\label{one}
\end{theorem}

\section{Axioms for Helping Centralities.} 

The Shapley value has a nice axiomatic characterization \cite{shapley1953value}. The axiomatic approach to characterizing various coalitional measures has developed into an important direction in coalitional game theory, and has recently been adapted to centrality measures as well. A natural question is whether our Helping Shapley value has a similar axiomatic characterization. 

To attempt such a characterization we define a number of properties reminiscent of the axiomatic characterization of the ordinary Shapley value:  
 
\begin{definition} Given graph $G=(V,E)$, a
function $f:\Gamma[V] \rightarrow \mathbb{R}^{V}$ satisfies the axioms of
\begin{itemize}
\item[-] \textbf{linearity} if for every player $x\in V$ and any two games $v_{1},v_{2}$ on $V$, $f[v_{1}+v_{2}](x)= f[v_{1}](x)+f[v_{2}](x)$ and, for every $\alpha \in \mathbb{R}$, 
$f[\alpha v_{1}](x) = \alpha\cdot f[v_{1}](x)$. 
\item[-] \textbf{null helping} if for every game $v$ on $V$ and $i\in V$ s.t. for every $S\subseteq V$, $v(S\cup \{i\}\cup N(i))=v(S)$ then $f[v](i)=0$. 
\item[-] \textbf{veto game symmetry} if for any veto game $v_{S}$ and players $x,y$ , $f[v_{S}](x),f[v_{S}](y)>0 \Rightarrow f[v_{S}](x)=f[v_{S}](y).$
\end{itemize} 
\end{definition} 

\begin{theorem} 
The Helping Shapley value satisfies the linearity and null helping axioms.
\label{axioms} 
\end{theorem}

Unfortunately while the Shapley value satisfies veto game symmetry, this is \textbf{not} true for the Helping Shapley value: 

\begin{example} 
Consider the star network $S_{n}$ in figure~\ref{p2} (b). Then in the unanimity game $v_{N}$ on $S_{n}$ (corresponding to $S=\{1,2,\ldots, n\}$) we have $
HSh[v_{S}](1)=1 > 0,$ $
HSh[v_{S}](i)= \frac{1}{n-1} >0  \mbox{ for all } i=2,\ldots, n. $
Indeed, node $i\geq 2$ is pivotal for $\pi$ iff all other nodes in $S\setminus \{1,i\}$ appear before $i$ in $\pi$. This happens with probability $1/(n-1)$. 
\end{example}  

This mismatch has implications for the axiomatic characterization of the Helping Shapley value: for the ordinary Shapley value its uniqueness amounts to establishing veto game symmetry, which normally follows from an \textit{equal treatment} axiom. The lack of veto symmetry means that we cannot adapt the classical proof of the Uniqueness of the Shapley value to the Helping Shapley value, but we only have the following weaker version: 

\begin{theorem} 
The Helping Shapley value $HSh$ is the only function $f$ that satisfies linearity  and $
f[v_{S}]=HSh[v_{S}]$ for all veto games $v_{S}$. 
\label{axiomatic}
\end{theorem}

\section{Related work\protect\footnote{F\lowercase{or reasons of space this section gives only minimal pointers to the existing literature;  A comprehensive treatment is deferred to the full (journal) version.}}}

Our work combines several important lines of research: the extensive literature on (game-theoretic) centrality measures (see \cite{bloch2017centrality,tarkowski2017game} for reviews from different perspectives) and that on compact representation frameworks for cooperative games \cite{deng1994complexity,conitzer2003complexity,ieong2005marginal,elkind2007computational,bachrach2008coalitional}. Through our Theorem~\ref{axiomatic} we connect to the growing literature 
on the axiomatic characterization of  values and centrality measures (for a recent example see \cite{skibski2017axiomatic}).  

Ideas related to the use of compact representations in defining notions of network centrality have been considered (implicitly or explicitly) in previous literature, e.g. \cite{tarkowski2018efficient,skibski2017axiomatic}. This last paper is, perhaps, the closest in spirit to our approach. They undertake a comprehensive study of classes of network centralities and identify axiomatic foundations for various representational frameworks. Compared to this work our focus is, however, different: we strive to include capabilities/tasks explicitly into the representational framework, and identify one framework which does just that. 

Finally, a related problem (but with different technical concerns) is \textit{team formation} in the knowledge discovery literature \cite{lappas2009finding,li2010team}.

\section{Conclusions, Further Work, Open Problems.}

Our work provides two important conceptual contributions: 
\begin{itemize} 
\item[
\noindent (1).] Giving \textit{an explicit framework for representing capabilities to perform tasks} in measures of network centrality, and
\item[
\noindent (2).] Proposing \textit{the new notion(s) of helping centrality}. We have given two such measures, which perform similarly on the 9/11 network.  Helping Centrality seems to be slightly more discriminating (and has sometimes exact formulas) but seems to lack "nice" axiomatizations. 
\end{itemize} 
Several open issues arise: first of all, we have only used one of the several formalisms for CSG games. A more extensive investigation of the representational power of (other) families would be in order. So would the computational and experimental aspects of helping centralities. The problem of defining representational formalisms that are able to naturally model all "reasonable" centrality measures is, we feel, an interesting one. Finally, note that Helping Centrality is missing from the list of $\#P$-complete results of Theorem~\ref{sharp-p}. We leave its complexity open. 

\bibliographystyle{unsrt}
\bibliography{/Users/gistrate/Dropbox/texmf/bibtex/bib/bibtheory}

\begin{thebibliography}{10}

\bibitem{bloch2017centrality}
Francis Bloch, Matthew~O Jackson, and Pietro Tebaldi.
\newblock Centrality measures in networks.
\newblock {\em Available at SSRN 2749124}, 2017.

\bibitem{suri2008determining}
N~Suri and Y.~Narahari.
\newblock Determining the top-k nodes in social networks using the {S}hapley
  value.
\newblock In {\em Proceedings of AAMAS'08}, pages 1509--1512, 2008.

\bibitem{narayanam2014shapley}
Ramasuri Narayanam, Oskar Skibski, Hemank Lamba, and Tomasz Michalak.
\newblock A {S}hapley value-based approach to determine gatekeepers in social
  networks with applications.
\newblock In {\em Proceedings of the Twenty-First European Conference on
  Artificial Intelligence (ECAI'14)}, pages 651--656. IOS Press, 2014.

\bibitem{maghami2012identifying}
Mahsa Maghami and Gita Sukthankar.
\newblock Identifying influential agents for advertising in multi-agent
  markets.
\newblock In {\em Proceedings of the 11th International Conference on
  Autonomous Agents and Multiagent Systems-Volume 2}, pages 687--694.
  International Foundation for Autonomous Agents and Multiagent Systems, 2012.

\bibitem{lindelauf2013cooperative}
R.~Lindelauf, H.~Hamers, and B.~Husslage.
\newblock Cooperative game theoretic centrality analysis of terrorist networks:
  The cases of jemaah islamiyah and al qaeda.
\newblock {\em European Journal of Operational Research}, 229(1):230--238,
  2013.

\bibitem{michalak2015defeating}
Tomasz~P. Michalak, Talal Rahwan, Oskar Skibski, and Michael Wooldridge.
\newblock Defeating terrorist networks with game theory.
\newblock {\em IEEE Intelligent Systems}, 30(1):53--61, 2015.

\bibitem{szczepanski2014centrality}
Piotr~L Szczepanski, Tomasz~P Michalak, and Michael~J Wooldridge.
\newblock A centrality measure for networks with community structure based on a
  generalization of the {O}wen value.
\newblock In {\em ECAI}, volume~14, pages 867--872. Citeseer, 2014.

\bibitem{tarkowski2017game}
Mateusz~K. Tarkowski, Tomasz~P. Michalak, Talal Rahwan, and Michael Wooldridge.
\newblock Game-theoretic network centrality: A review.
\newblock {\em arXiv preprint arXiv:1801.00218}, 2017.

\bibitem{tarkowski2016closeness}
Mateusz~K Tarkowski, Piotr Szczepa{\'n}ski, Talal Rahwan, Tomasz~P Michalak,
  and Michael Wooldridge.
\newblock Closeness centrality for networks with overlapping community
  structure.
\newblock In {\em Thirtieth AAAI Conference on Artificial Intelligence}, 2016.

\bibitem{aadithya2010game}
Karthik~V. Aadithya and Balaraman Ravindran.
\newblock Game theoretic network centrality: exact formulas and efficient
  algorithms.
\newblock In {\em Proceedings of the 2010 International Conference on
  Autonomous Agents and Multiagent Systems: volume 1-Volume 1}, pages
  1459--1460, 2010.

\bibitem{szczepanski2016efficient}
Piotr~L Szczepa{\'n}ski, Tomasz~P Michalak, and Talal Rahwan.
\newblock Efficient algorithms for game-theoretic betweenness centrality.
\newblock {\em Artificial Intelligence}, 231:39--63, 2016.

\bibitem{tarkowski2018efficient}
Mateusz~K Tarkowski, Piotr~L Szczepa{\'n}ski, Tomasz~P Michalak, Paul
  Harrenstein, and Michael Wooldridge.
\newblock Efficient computation of semivalues for game-theoretic network
  centrality.
\newblock {\em Journal of Artificial Intelligence Research}, 63:145--189, 2018.

\bibitem{granovetter1977strength}
Mark~S Granovetter.
\newblock The strength of weak ties.
\newblock In {\em Social networks}, pages 347--367. Elsevier, 1977.

\bibitem{skibski2017axiomatic}
Oskar Skibski, Tomasz~P Michalak, and Talal Rahwan.
\newblock Axiomatic characterization of game-theoretic network centralities.
\newblock In {\em Thirty-First AAAI Conference on Artificial Intelligence},
  2017.

\bibitem{deng1994complexity}
X.~Deng and C.H. Papadimitriou.
\newblock On the complexity of cooperative solution concepts.
\newblock {\em Mathematics of Operations Research}, pages 257--266, 1994.

\bibitem{conitzer2003complexity}
Vincent Conitzer and Tuomas Sandholm.
\newblock Complexity of determining nonemptiness of the core.
\newblock In {\em Proceedings of the 4th ACM Conference on Electronic
  Commerce}, pages 230--231. ACM, 2003.

\bibitem{ieong2005marginal}
Samuel Ieong and Yoav Shoham.
\newblock Marginal contribution nets: a compact representation scheme for
  coalitional games.
\newblock In {\em Proceedings of the 6th ACM conference on Electronic
  commerce}, pages 193--202. ACM, 2005.

\bibitem{elkind2007computational}
Edith Elkind, Leslie~Ann Goldberg, Paul Goldberg, and Michael Wooldridge.
\newblock Computational complexity of weighted threshold games.
\newblock In {\em Proc. AAAI'07}, volume~22, page 718, 2007.

\bibitem{bachrach2008coalitional}
Yoram Bachrach and Jeffrey~S Rosenschein.
\newblock Coalitional skill games.
\newblock In {\em Proceedings of the 7th international joint conference on
  Autonomous agents and multiagent systems (AAMAS'08)}, pages 1023--1030, 2008.

\bibitem{chalkiadakis2011computational}
G.~Chalkiadakis, E.~Elkind, and M.~Wooldridge.
\newblock Computational aspects of cooperative game theory.
\newblock {\em Synthesis Lectures on Artificial Intelligence and Machine
  Learning}, 2011.

\bibitem{karlin2017game}
Anna~R Karlin and Yuval Peres.
\newblock {\em Game theory, alive}, volume 101.
\newblock American Mathematical Soc., 2017.

\bibitem{bachrach2013computing}
Y.~Bachrach, D.~Parkes, and J.~Rosenschein.
\newblock Computing cooperative solution concepts in coalitional skill games.
\newblock {\em Artificial Intelligence}, 204:1--21, 2013.

\bibitem{krebs2002mapping}
Valdis~E Krebs.
\newblock Mapping networks of terrorist cells.
\newblock {\em Connections}, 24(3):43--52, 2002.

\bibitem{amer2004connectivity}
Rafael Amer and Jos{\'e}~Miguel Gim{\'e}nez.
\newblock A connectivity game for graphs.
\newblock {\em Mathematical {M}ethods of {O}perations {R}esearch},
  60(3):453--470, 2004.

\bibitem{myerson1977graphs}
Roger~B Myerson.
\newblock Graphs and cooperation in games.
\newblock {\em Mathematics of {O}perations {R}esearch}, 2(3):225--229, 1977.

\bibitem{skibski2019enumerating}
Oskar Skibski, Talal Rahwan, Tomasz~P Michalak, and Michael Wooldridge.
\newblock Enumerating connected subgraphs and computing the myerson and shapley
  values in graph-restricted games.
\newblock {\em ACM Transactions on Intelligent Systems and Technology (TIST)},
  10(2):15, 2019.

\bibitem{dubey1981value}
Pradeep Dubey, Abraham Neyman, and Robert~James Weber.
\newblock Value theory without efficiency.
\newblock {\em Mathematics of Operations Research}, 6(1):122--128, 1981.

\bibitem{michalak2013efficient}
Tomasz~P Michalak, Karthik~V Aadithya, Piotr~L Szczepanski, Balaraman
  Ravindran, and Nicholas~R Jennings.
\newblock Efficient computation of the shapley value for game-theoretic network
  centrality.
\newblock {\em Journal of Artificial Intelligence Research}, pages 607--650,
  2013.

\bibitem{aziz2009algorithmic}
Haris Aziz.
\newblock {\em Algorithmic and complexity aspects of simple coalitional games}.
\newblock PhD thesis, University of Warwick, 2009.

\bibitem{alshebli2019measure}
Bedoor~K Alshebli, Tomasz~P Michalak, Oskar Skibski, Michael Wooldridge, and
  Talal Rahwan.
\newblock A measure of added value in groups.
\newblock {\em ACM Transactions on Autonomous and Adaptive Systems (TAAS)},
  13(4):18, 2019.

\bibitem{shapley1953value}
Lloyd~S Shapley.
\newblock A value for n-person games.
\newblock {\em Contributions to the Theory of Games}, 2(28):307--317, 1953.

\bibitem{lappas2009finding}
Theodoros Lappas, Kun Liu, and Evimaria Terzi.
\newblock Finding a team of experts in social networks.
\newblock In {\em Proceedings of the 15th ACM SIGKDD international conference
  on Knowledge discovery and data mining}, pages 467--476. ACM, 2009.

\bibitem{li2010team}
Cheng-Te Li and Man-Kwan Shan.
\newblock Team formation for generalized tasks in expertise social networks.
\newblock In {\em 2010 IEEE Second International Conference on Social
  Computing}, pages 9--16. IEEE, 2010.

\end{thebibliography}


\begin{thebibliography}{}

\bibitem[\protect\citeauthoryear{Aadithya and
  Ravindran}{2010}]{aadithya2010game}
Karthik~V. Aadithya and Balaraman Ravindran.
\newblock Game theoretic network centrality: exact formulas and efficient
  algorithms.
\newblock In {\it  Proceedings of the 2010 International Conference on
  Autonomous Agents and Multiagent Systems: volume 1-Volume 1}, pages
  1459--1460, 2010.

\bibitem[\protect\citeauthoryear{Aziz}{2009}]{aziz2009algorithmic}
Haris Aziz.
\newblock {\it  Algorithmic and complexity aspects of simple coalitional games}.
\newblock PhD thesis, University of Warwick, 2009.

\bibitem[\protect\citeauthoryear{Bachrach and
  Rosenschein}{2008}]{bachrach2008coalitional}
Yoram Bachrach and Jeffrey~S Rosenschein.
\newblock Coalitional skill games.
\newblock In {\it  Proceedings of the 7th international joint conference on
  Autonomous agents and multiagent systems (AAMAS'08)}, pages 1023--1030, 2008.

\bibitem[\protect\citeauthoryear{Bachrach \bgroup \it  et al.\egroup
  }{2008}]{bachrach2008approximating}
Yoram Bachrach, Evangelos Markakis, Ariel~D Procaccia, Jeffrey~S Rosenschein,
  and Amin Saberi.
\newblock Approximating power indices.
\newblock In {\it  Proceedings of the 7th international joint conference on
  Autonomous agents and multiagent systems-Volume 2}, pages 943--950.
  International Foundation for Autonomous Agents and Multiagent Systems, 2008.

\bibitem[\protect\citeauthoryear{Bachrach \bgroup \it  et al.\egroup
  }{2013}]{bachrach2013computing}
Y.~Bachrach, D.~Parkes, and J.~Rosenschein.
\newblock Computing cooperative solution concepts in coalitional skill games.
\newblock {\it  Artificial Intelligence}, 204:1--21, 2013.

\bibitem[\protect\citeauthoryear{Bloch \bgroup \it  et al.\egroup
  }{2017}]{bloch2017centrality}
Francis Bloch, Matthew~O Jackson, and Pietro Tebaldi.
\newblock Centrality measures in networks.
\newblock {\it  Preprint, available at SSRN 2749124}, 2017.

\bibitem[\protect\citeauthoryear{Chalkiadakis \bgroup \it  et al.\egroup
  }{2011}]{chalkiadakis2011computational}
G.~Chalkiadakis, E.~Elkind, and M.~Wooldridge.
\newblock Computational aspects of cooperative game theory.
\newblock {\it  Synthesis Lectures on Artificial Intelligence and Machine
  Learning}, 2011.

\bibitem[\protect\citeauthoryear{Conitzer and
  Sandholm}{2003}]{conitzer2003complexity}
Vincent Conitzer and Tuomas Sandholm.
\newblock Complexity of determining nonemptiness of the core.
\newblock In {\it  Proceedings of the 4th ACM Conference on Electronic
  Commerce}, pages 230--231. ACM, 2003.

\bibitem[\protect\citeauthoryear{Deng and
  Papadimitriou}{1994}]{deng1994complexity}
Xiaotie Deng and Christos~H. Papadimitriou.
\newblock On the complexity of cooperative solution concepts.
\newblock {\it  Mathematics of Operations Research}, 19(2):257--266, 1994.

\bibitem[\protect\citeauthoryear{Dubey \bgroup \it  et al.\egroup
  }{1981}]{dubey1981value}
Pradeep Dubey, Abraham Neyman, and Robert~James Weber.
\newblock Value theory without efficiency.
\newblock {\it  Mathematics of Operations Research}, 6(1):122--128, 1981.

\bibitem[\protect\citeauthoryear{Elkind \bgroup \it  et al.\egroup
  }{2007}]{elkind2007computational}
Edith Elkind, Leslie~Ann Goldberg, Paul Goldberg, and Michael Wooldridge.
\newblock Computational complexity of weighted threshold games.
\newblock In {\it  Proceedings of the national conference on artificial
  intelligence}, volume~22, page 718. Menlo Park, CA; Cambridge, MA; London;
  AAAI Press; MIT Press; 1999, 2007.

\bibitem[\protect\citeauthoryear{Granovetter}{1977}]{granovetter1977strength}
Mark~S Granovetter.
\newblock The strength of weak ties.
\newblock In {\it  Social networks}, pages 347--367. Elsevier, 1977.

\bibitem[\protect\citeauthoryear{Ieong and Shoham}{2005}]{ieong2005marginal}
Samuel Ieong and Yoav Shoham.
\newblock Marginal contribution nets: a compact representation scheme for
  coalitional games.
\newblock In {\it  Proceedings of the 6th ACM conference on Electronic
  commerce}, pages 193--202. ACM, 2005.

\bibitem[\protect\citeauthoryear{Karlin and Peres}{2017}]{karlin2017game}
Anna~R Karlin and Yuval Peres.
\newblock {\it  Game theory, alive}, volume 101.
\newblock American Mathematical Soc., 2017.

\bibitem[\protect\citeauthoryear{Karpov}{2014}]{karpov2014equal}
Alexander Karpov.
\newblock Equal weights coauthorship sharing and the shapley value are
  equivalent.
\newblock {\it  Journal of Informetrics}, 8(1):71--76, 2014.

\bibitem[\protect\citeauthoryear{Krebs}{2002}]{krebs2002mapping}
Valdis~E Krebs.
\newblock Mapping networks of terrorist cells.
\newblock {\it  Connections}, 24(3):43--52, 2002.

\bibitem[\protect\citeauthoryear{Lindelauf \bgroup \it  et al.\egroup
  }{2013}]{lindelauf2013cooperative}
RHA Lindelauf, HJM Hamers, and BGM Husslage.
\newblock Cooperative game theoretic centrality analysis of terrorist networks:
  The cases of jemaah islamiyah and al qaeda.
\newblock {\it  European Journal of Operational Research}, 229(1):230--238,
  2013.

\bibitem[\protect\citeauthoryear{Maghami and
  Sukthankar}{2012}]{maghami2012identifying}
Mahsa Maghami and Gita Sukthankar.
\newblock Identifying influential agents for advertising in multi-agent
  markets.
\newblock In {\it  Proceedings of the 11th International Conference on
  Autonomous Agents and Multiagent Systems-Volume 2}, pages 687--694.
  International Foundation for Autonomous Agents and Multiagent Systems, 2012.

\bibitem[\protect\citeauthoryear{Michalak \bgroup \it  et al.\egroup
  }{2013}]{michalak2013efficient}
Tomasz~P Michalak, Karthik~V Aadithya, Piotr~L Szczepanski, Balaraman
  Ravindran, and Nicholas~R Jennings.
\newblock Efficient computation of the shapley value for game-theoretic network
  centrality.
\newblock {\it  Journal of Artificial Intelligence Research}, 46:607--650, 2013.

\bibitem[\protect\citeauthoryear{Michalak \bgroup \it  et al.\egroup
  }{2015}]{michalak2015defeating}
Tomasz~P Michalak, Talal Rahwan, Oskar Skibski, and Michael Wooldridge.
\newblock Defeating terrorist networks with game theory.
\newblock {\it  IEEE Intelligent Systems}, 30(1):53--61, 2015.

\bibitem[\protect\citeauthoryear{Narayanam \bgroup \it  et al.\egroup
  }{2014}]{narayanam2014shapley}
Ramasuri Narayanam, Oskar Skibski, Hemank Lamba, and Tomasz Michalak.
\newblock A shapley value-based approach to determine gatekeepers in social
  networks with applications.
\newblock In {\it  Proceedings of the Twenty-first European Conference on
  Artificial Intelligence}, pages 651--656. IOS Press, 2014.

\bibitem[\protect\citeauthoryear{Shapley}{1953}]{shapley1953value}
Lloyd~S Shapley.
\newblock A value for n-person games.
\newblock {\it  Contributions to the Theory of Games}, 2(28):307--317, 1953.

\bibitem[\protect\citeauthoryear{Skibski \bgroup \it  et al.\egroup
  }{2017}]{skibski2017axiomatic}
Oskar Skibski, Tomasz~P Michalak, and Talal Rahwan.
\newblock Axiomatic characterization of game-theoretic network centralities.
\newblock In {\it  AAAI}, pages 698--705, 2017.

\bibitem[\protect\citeauthoryear{Suri and Narahari}{2008}]{suri2008determining}
N~Rama Suri and Yadati Narahari.
\newblock Determining the top-k nodes in social networks using the shapley
  value.
\newblock In {\it  Proceedings of the 7th international joint conference on
  Autonomous agents and multiagent systems-Volume 3}, pages 1509--1512.
  International Foundation for Autonomous Agents and Multiagent Systems, 2008.

\bibitem[\protect\citeauthoryear{Szczepanski \bgroup \it  et al.\egroup
  }{2014}]{szczepanski2014centrality}
Piotr~L Szczepanski, Tomasz~P Michalak, and Michael Wooldridge.
\newblock A centrality measure for networks with community structure based on a
  generalization of the owen value.
\newblock In {\it  ECAI}, volume~14, pages 867--872. Citeseer, 2014.

\bibitem[\protect\citeauthoryear{Szczepa{\'n}ski \bgroup \it  et al.\egroup
  }{2016}]{szczepanski2016efficient}
Piotr~L Szczepa{\'n}ski, Tomasz~P Michalak, and Talal Rahwan.
\newblock Efficient algorithms for game-theoretic betweenness centrality.
\newblock {\it  Artificial Intelligence}, 231:39--63, 2016.

\bibitem[\protect\citeauthoryear{Tarkowski \bgroup \it  et al.\egroup
  }{2016}]{tarkowski2016closeness}
Mateusz~K Tarkowski, Piotr Szczepa{\'n}ski, Talal Rahwan, Tomasz~P Michalak,
  and Michael Wooldridge.
\newblock Closeness centrality for networks with overlapping community
  structure.
\newblock In {\it  Thirtieth AAAI Conference on Artificial Intelligence}, 2016.

\bibitem[\protect\citeauthoryear{Tarkowski \bgroup \it  et al.\egroup
  }{2017}]{tarkowski2017game}
Mateusz~K. Tarkowski, Tomasz~P. Michalak, Talal Rahwan, and Michael Wooldridge.
\newblock Game-theoretic network centrality: A review.
\newblock {\it  arXiv preprint arXiv:1801.00218}, 2017.

\bibitem[\protect\citeauthoryear{Tarkowski \bgroup \it  et al.\egroup
  }{2018}]{tarkowski2018efficient}
Mateusz~K Tarkowski, Piotr~L Szczepa{\'n}ski, Tomasz~P Michalak, Paul
  Harrenstein, and Michael Wooldridge.
\newblock Efficient computation of semivalues for game-theoretic network
  centrality.
\newblock {\it  Journal of Artificial Intelligence Research}, 63:145--189, 2018.

\end{thebibliography}

\newpage
\section*{Supplemental Material} 

In this manuscript we give some extra details on the proofs omitted from the main document: 

\section{Proof of Theorem~\ref{rational-csg}}

We will use the decomposition of coalitional skill games by linearity to reduce the problem to reasoning about \textit{simple skill games}, i.e. CSG consisting of a single task $t$ of unit profit. In this setting a  coalition $S$ is called {\it  winning} if it can accomplish task $t$ and {\it  losing} otherwise. Define by $W(t)$ the set of minimal winning coalitions for task $t$, i.e. of subsets $A$ such that $v(A)=1$ and $v(B)=0$ for any strict subset $B$ of $A$. 

Then the polynomial 
\[
1-\prod_{A \in W(t)} (1-\prod_{r\in A} Y_{r})
\]
is equal to 1 precisely when for some $A \in W(t)$ we have $Y_{r}=1$ for all $r\in A$.  Hence, denoting $S=\{r:Y_{r}=1\}$, we have 
\[
v(S)=1-\prod_{A \in W(t)} (1-\prod_{r\in A} Y_{r})
\]

\section{Proof of Theorem~\ref{pc}}

We will again use the decomposition of coalitional skill games by linearity to reduce the problem to reasoning about \textit{simple skill games}.

Clearly $\Phi_{i}^{\beta}$ can be computed by estimating, for $k=0,\ldots, n-1$ quantities $\mathbb{E}_{C\in C_{k}}[MC(C,i)].$ Nonzero marginal contributions arise from 
 ordered coalitions $C\in C_{k}$ such that 
\begin{itemize} 
\item[]  (1). $i$ is the last element of $C$. 
\item[] (2). $C$ is winning. 
\item[] (3). $C\setminus \{i\}$ is {\bf not} winning. 
\end{itemize} 
For every task $T_{j}$ (identified with a multiset of skills) 
denote by $\mathcal{T}^{*}_{j}$ the set of submultisets of $T_{j}$. 
For $T\in \mathcal{T}^{*}_{j}$ and $1\leq r\leq n-1$ denote by $n_{T,r}^{i}$ the number of {\it  ordered coalitions} $X$ {\it  of size exactly $r$ and not containing $i$} such that $S_{X}\cap T_{j}=T$. Denote $W_{i,j}=C_{i}\cap T_{j}$
the set of skills of agent $i$ that can contribute towards completing task $j$.
An ordered coalition satisfies properties (1)-(3) above {\bf if and only if}
$T_{j}\setminus W_{i,j} \subseteq S_{C\setminus \{i\}} \subsetneq T_{j}$. 
Thus 
\begin{equation}
\Phi_{i}^{\beta}= \sum\limits_{r=0}^{n-1} \beta(r)\big[\displaystyle\sum\limits_{T_{j}\setminus W_{i,j} \subseteq T \subsetneq T_{j}} \frac{(n-r-1)!}{n!}\cdot n_{T,r}^{i}\big].
\end{equation} 

We use a dynamic programming approach to compute parameters
$n_{T,r}^{i}$. The table has at most $2^{k}$ columns, each corresponding to a submultiset of $T_{j}$. Rows of the table correspond to pairs $(r,s)$, where 
$r\leq s\leq n$. The element on row $(r,s)$ and $T$, denoted by $n_{T,r,s}^{i}$, counts 
ordered coalitions $X$ of size $r$ not containing element $i$, formed with elements
from $a_1,a_2,\ldots, a_{s}$, such that $S_{X}\cap T_{j}=T$. 
Clearly $n_{T,r}^{i}=n_{T,r,n}^{i}.$

We start the table by filling in rows $(r,s)=(0,0)$ and
$(r,s)=(1,1)$. Clearly, $n_{T,0,k}^{i}=0$ for all $k\geq 1$ and $T
\subseteq T_{j}$, and $n_{T,1,1}^{i}$ equals the number of players $k
\neq i$ such that $S_{k}\cap T_{j}=T$. Thus rows $(0,0)$ and $(1,1)$
can be completed by simple player inspection. 

Now, coalitions $X$ of size $r$ not containing element $i$ with elements
from the set $a_{1},a_{2},\ldots a_{s}$ such that $S_{X}
\cap T_{j}=T$ decompose into two types: 
\begin{itemize} 
\item[-] Coalitions {\it  not containing} $a_{s}$. Their number is counted
  by $n_{T,r,s-1}^{i}$. 
\item[-] Coalitions {\it  containing} $a_{s}$. Let $Y=X\setminus a_{s}$,
  $D=S_{Y}\cap T_{j}$, $D\subseteq T$,$E=S_{a_s}\cap T_{j}$. It follows that
$D\cup E=T.$  Thus we get recurrence $n_{T,r,s}^{i}=n_{T,r,s-1}^{i}+\sum\limits_{D\cup (S_{s}\cap T_{j})=T} n_{D,r-1,s-1}^{i}$.  This equation allows us to fill row $(r,s)$ from rows $(r,s-1)$ and $(r-1,s-1)$, ultimately allowing to compute
parameters $n_{T,r}^{i}$ for all $i$. 
\end{itemize} 

It is easily seen that the complexity of the provided algorithm is $O(2^{k}\cdot poly(|\Gamma |))$. Thus the problem is fixed parameter tractable. 

\section{Proof of Theorem~\ref{positive-helping}}

By additivity it is enough to assume that $\Gamma$ is a single-task game. 

We will investigate quantity 
\begin{align} 
[v[S_{\pi}^{x}\cup \{x\}\cup N(v(S_{\pi}^{x}\cup \{x\})]-v(S_{\pi}^{x}\cup N(S_{\pi}^{x}))]- \nonumber \\
-\frac{1}{\Delta^{*}+1}\cdot (v(S_{\pi}^{x}\cup \{x\})-v(S_{\pi}^{x})). 
\label{fss}
\end{align} 

The average of this quantity over all permutations $\pi\in S_{n}$ yields $HC(x)$. Further define, for $A\subseteq V$, 
\begin{equation} 
h(A):= v(A\cup N(A))-\frac{1}{\Delta^{*}+1}\cdot v(A).
\end{equation} 

Quantity in Equation~(\ref{fss}) is, of course, nothing but $h(S_{\pi}^{x}\cup \{x\})-h(S_{\pi}^{x})$, and 
\begin{equation} 
HC(x)=E_{\pi\in S_{n}}[h(S_{\pi}^{x}\cup \{x\})-h(S_{\pi}^{x})]
\end{equation} 

\textbf{ Case A: $v(A)=0$.} Then $h(A)=0$ iff $v(A\cup N(A))=0$, $h(A)=1$, otherwise. 

\textbf{Case B: $v(A)=1$}.  Then $h(A)=\frac{\Delta^{*}}{\Delta^{*}+1}$. 

In conclusion:
\begin{enumerate} 
\item When $v(A)=1$ we also have $v(A\cup \{x\})=1$, hence $h(A\cup \{x\})=h(A)=\frac{\Delta^{*}}{\Delta^{*}+1}$ and $h(A\cup \{x\})-h(A)= 0$. 
\item If $v(A\cup \{x\})=0$ and $v(A\cup \{x\}\cup N(A\cup \{x\}))=0$, then $h(A\cup \{x\})-h(A)= 0$. 
 \item If $v(A\cup \{x\})=0$ and $v(A\cup \{x\}\cup N(A\cup \{x\}))=1$ but $v(A\cup N(A))=0$, then $h(A\cup \{x\})-h(A)= 1$. 
\item If $v(A\cup \{x\})=1$ and $v(A\cup N(A))=0$ then $h(A\cup \{x\})-h(A)= \frac{\Delta^{*}}{\Delta^{*}+1}$.
\item The only case when $h(A\cup \{x\})-h(A)$ is negative is when $v(A\cup N(A))=v(A\cup \{x\})=1$, $v(A)=0$, in which case $h(A\cup \{x\})-h(A)= -\frac{1}{\Delta^{*}+1}$. 
\end{enumerate} 

To compute $HC(x)$ one averages (over all $\pi\in S_{n}$) both positive and negative terms of type $h(S_{\pi}^{x}\cup \{x\})-h(S_{\pi}^{x})$.

If $x$ has no skill useful to the unique task $t$ then Case 5 cannot happen (since condition $v(A\cup \{x\})=1, v(A)=0$ is impossible), and we are done. 

In the opposite situation, suppose $x$ has the useful skill $s$.   To prove that $HC(x)\geq 0$ it is enough to compare the contribution of positive terms (Cases 3,4 in the previous enumeration) with that of negative terms (Case 5).  It turns out that we will only need to compare contributions from Cases 4 and 5. 

In both cases 4 and 5, as $v(S_{\pi}^{x}\cup \{x\})=1$, $v(S_{\pi}^{x})=0$ and $x$ adds one more skill to $S_{\pi}^{x}$, $x$ must add precisely the unique copy of skill $s$ that $S_{\pi}^{x}$ is missing to complete task $t$. 


Since $v(S_{\pi}^{x}\cup \{x\})=1$, $v(S_{\pi}^{x})=0$, in both cases $x$ must be the $k_{s,t}$'th element with skill $s$ in $\pi$. The difference between Cases 4 and 5 is, therefore, whether 
$v(S_{\pi}^{x}\cup N(S_{\pi}^{x})=1$, that is whether some node in $S_{\pi}^{x}$ is adjacent (or not) to a node with skill $s$ in $V\setminus S_{\pi}^{x}$. 

Let's condition on $(V\setminus (S_{\pi}^{x})\cap P(s)$ being equal to a fixed set $C$ (of size $n_{s}-k_{s}+1$) and compare the contributions of permutations falling in Cases 4 and 5 to $HC[x]$. 

It is easy to see that $v(S_{\pi}^{x}\cup N(S_{\pi}^{x}))=1$ iff some node in $S_{\pi}^{x}$ is adjacent to some node in $C$, i.e. if $S_{\pi}^{x}$ contains some node in $N(C)\setminus C$. 

In other words, we are in Case 4 iff $C\cup N(C)\subseteq V\setminus S_{\pi}^{x}$, i.e. when $x$ is the first element in $\{x\}\cup  N(C)\setminus  C$. This happens with probability $\frac{1}{|N(C)\setminus C|+1}$. The contribution of permutations $\pi$ with $S_{\pi}^{x}$ falling in Case 4 to the sum is $\frac{\Delta^{*}}{(\Delta^{*}+1)(|N(C)\cup C|+1)}$. 

As for case 5, we are in that case with conditional probability $\frac{|N(C)\setminus C|}{|N(C)\setminus C|+1}$. The contribution of permutations $\pi$ with $S_{\pi}^{x}$ falling in Case 5 to the sum is $-\frac{|N(C)\setminus  C|}{(\Delta^{*}+1)(|N(C)\setminus C|+1)}$.

The difference in contributions is $\frac{1}{(\Delta^{*}+1)(|N(C)\setminus C|+1)}\cdot [\Delta^{*}-|N(C)\setminus C|]$, which is $\geq 0$ by the definition of $\Delta^{*}$.

\section{Proof of Theorem~\ref{seven}}

By the linear decomposition of games as combination of veto games, we only need to compute the Helping Shapley value of $i$ for the $S$-veto game $v_{S}$. 
Clearly, if $i\not \in S\cup N(S)$, then $HSh[v_{S}](i)=0$. Indeed, in this case $i$ can help no coalition $T$ contain $S$, since $\hat{N}(i)\cap S= \emptyset$). So we concentrate on the case $i\in S\cup N(S)$.  There are two subcases: 
\begin{itemize}
\item[-] $i\in S$. Then $i$ is pivotal iff all elements of $NC(i,S)$ appear before $i$ in $\pi$. The probability of this happening is $\frac{1}{|NC(i,S)|+1}$.
\item[-] $i\in N(S)\setminus S$. Then $i$ is pivotal iff all elements of $NC(i,S)$ appear before $i$ in $\pi$ and some element of $S\cap N(i)$ appears after $i$ in $\pi$. That is 
\begin{itemize} 
\item[(a). ] \textbf{all elements of $NC(i,S)$ apppear before $i$ in $\pi$} (an event which happens with probability $1/(|NC(i,S)|+1)$),  but it is not the case that 
\item[(b). ] \textbf{$i$ is the last element of $S\cap N(i)$ in $\pi$.} Events (a), (b) are independent, since they refer to sets ($NC(i,S)$, $S\cap N(i)$ that do not intersect), and the probability of (b) happening is 
$1/|Cov(i,S)|$.
\end{itemize} 
\end{itemize} 

To derive the second formula we simply use the formula for the Shapley value. 

\section{Proof of Theorem~\ref{sharp-p}}

Since computing the Shapley value of an arbitrary $CSG$ is $\#P$-complete \cite{aziz2009algorithmic}, the result follows by chosing $G=(V,\emptyset)$, the empty graph on $V$. In such a case $v(S)=v_{cen}(S)$ for every $S\subseteq V$. 
 Also $Sh[v](x)=HSh[v](x)$, $HC(x)=Sh[v](x)\cdot (1-\frac{1}{\Delta+1})$. 

\section{Proof of Theorem~\ref{one}}

Decomposing again the game $v$ into a weighted combination of single task games, we reduce the proof of the formulas to the setting when the game is a such a game. 

A coalition $S$ satisfies $v_{cen}(S)=1$ iff $v(S\cup N(S))=1$. 
We infer that $x$ is pivotal to coalition $S$ under $v_{cen}$ iff $v(S\cup \{x\}\cup N(S\cup \{x\})=1$ , but $v(S\cup N(S))=0$.  \textit{Because the game is a pure skill game}, these conditions are equivalent to requiring that $x$ or some node in $N(x)$ can perform task $t$, but  nodes in $S$ cannot. In other words $x$ is the first node in permutation $\pi$ from the set $P(t)\cup N(P(t))$.  

Given this argument, the first proof is simple: The probability that in a random permutation  $x$ is the first player among those in $P(t)\cup N(P(t))$ is $1/|P(t) \cup N(P(t))|$. A similar computation works for computing the Shapley value of $v$, establishing the formula for $HC[v](x)$. Note that, since $|P(t)\cup N(P(t))|\leq (\Delta + 1)|P(t)|$, $HC[v](x)\geq 0$. 

As for the second proof, for $x$ to be pivotal to $HSh[v]$, none of the elements before $x$ in $\pi$ must be in $T$, while $x$ must be in $T\cup N(T)$. There are two cases: 
\begin{itemize} 
\item[-] $x\in P(t)$, i.e. $t\in T_{x}$.  Then $x$ must be the first element of $P(t)$ in permutation $\pi$. In a random permutation this happens with probability $1/|P(t)|$. 
\item[-] $x\in N(P(t))\setminus P(t)$, i.e. $t\in T_{\hat{N}(x)}\setminus T_{x}$. Then $x$ must be the first element from $\{x\}\cup P(t)$ in permutation $\pi$. In a random permutation this happens with probability $1/(|P(t)|+1)$.
\end{itemize}  
The computation of the Shapley value is equally simple, and uses similar probabilistic arguments. 

\section{Proof of Theorem~\ref{axioms}}

A couple of straightforward verifications. 

\section{Proof of Theorem~\ref{axiomatic}}

Recall \cite{karlin2017game} that any characteristic function can be written as a linear combination of the characteristic functions $v_{S}$ of all $2^{n}-1$ nontrivial $S$-veto games, $\Gamma=\sum\limits_{S\neq \emptyset} a_{S} v_{S}$. 

Now we apply linearity and we infer that 
\[
f[\Gamma]=f[\sum\limits_{S\neq \emptyset} a_{S} v_{S}]=\sum\limits_{S\neq \emptyset} a_{S} f[v_{S}] = \sum\limits_{S\neq \emptyset} a_{S} HSh[v_{S}] 
\]
Plugging in the explicit formula for $HSh[v_{S}]$ we immediately infer that $f[\Gamma]=HSh[\Gamma]$. 

\section{Helping Centralities in the 9/11 Terrorist Network: Some Details}

We have implemented a simple Python script \textit{calcShapley.py} for computing measures $Help$ and $HC$ (see the details on it in the next section).

The (raw) results are presented in Tables~\ref{default} and~\ref{default1}. 

\begin{table}[htp]
\caption{(Unnormalized) Helping Centralities.}
\begin{center}
\begin{tabular}{|c|c|}
\hline \hline
Nodes & Help \\
\hline 
16 & 0.279120879121 \\
\hline
13 & 0.234188034188\\
\hline
7,9,12 & 0.2\\
\hline
6,8,10,11,4 & 0.123076923077 \\
\hline
5 & 0.0888888 \\
\hline 
3,14,15 & 0.0791208791209\\
\hline 
1,19 & 0.0449328449328\\
\hline
17,18 & 0.034188034188\\
\hline 
2 & 0 \\
\hline
\end{tabular}
\end{center}
\label{default}
\end{table}%

\begin{table}[htp]
\caption{Shapley-Based Helping Centralities.}
\begin{center}
\begin{tabular}{|c|c|}
\hline \hline
Nodes & HC \\
\hline 
16 & 0.118687 \\
\hline
13 & 0.113248 \\
\hline
7,9,12 & 0.108198\\
\hline
6,8,10,11 & 0.052642 \\
\hline
4 & 0.049077 \\
\hline 
5 & 0.044026 \\
\hline 
3,14,15,1,19 & 0.019514\\
\hline 
17,18,2 & 0.014075 \\
\hline 
\end{tabular}
\end{center}
\label{default1}
\end{table}%

 \section{Details on the Python script \textit{calcShapley.py}} 

The input is read from a file and has the following structure: 
\begin{itemize} 
\item[-] On the first line there are 4 values: the number of terrorists, martial artists, pilots and edges in the network. The values are given in this order
\item[-] The second line names which nodes are martial artists
\item[-] The third line names which nodes are pilots
\item[-] The fourth line contains the number of pilots and the number of martial artists required to complete the attack
\item[-] The next lines, to the end of the file, contain two integers x and y representing an undirected edge between nodes x and y
\end{itemize} 
Following the input phase we construct a list "fact" to hold the factorials of numbers from 0 to 19.
Afterwards we can proceed with the actual computation.
The idea is to count how many coalition a certain node can help. We use backtracking to generate all unordered coalitions. If a coalition doesn't contain at least one pilot and two martial artists, we iterate through all 19 terrorists to see who could help this coalition achieve its goals. We perform three checks for all terrorists: does the coalition win if terrorist x joins it? Does it win if x joins with his neighbors? If the coalition with its neighbors doesn't win, does it win if x and his neighbors join? For the questions answered by yes, we count in the number of ordered coalitions helped in a certain way by x. To get the number of ordered coalitions from the unordered ones, we need to count all permutations that begin with the coalition in question, continue with terrorist x, and end with the rest of the terrorists that are not x and not in the coalition. To get this number, we multiply the factorial of the size of the coalition with the factorial of the remaining terrorists. From these counts we can then compute the Shapley value and the HC and Help indices introduced in the paper.
The function "iterate" recursively generates all combinations in a list "coalition", which contains 20 values of 0 or 1 representing for each terrorist whether he's in the coalition or not. In other words, if coalition[x] is 1, then x is in the current coalition. For each coalition and each terrorist x the check functions are called: checkshapley, checkhelp and checkvcen. These functions construct the coalition extended with x in the appropriate way according to the logic described above. The counts are held in the respective variables scoreshapley, scorehelp and scorevcen. Intuitively, the Shapley value for the game v is retrieved by dividing scoreshapley with the total number of coalitions (19!). The Shapley value in vcen is retrieved similarly by dividing scorevcen with 19!. Using these values we can compute the HC and Help indices described in the paper using their respective formulas.
In the end, the print$\_$all function is called. We print the Shapley, HC and Help values for each terrorist. The last three printed lines contain for each value concept the sum of values of all terrorists, verifying that the Shapley and HC values sum up to 1, whereas Help does not.

\end{document}